\documentclass[12pt]{article}

\usepackage[utf8]{inputenc}
\usepackage[english]{babel}

\usepackage{amsmath}
\usepackage{amsfonts}
\usepackage{amssymb}
\usepackage{amsthm}

\author{M.\,Ziatdinov}
\title{Quantum hashing based on symmetric groups}
\date{}

\newcommand{\KK}{\mathcal{K}}

\newcommand{\ket}[1]{|#1\rangle}

\newcommand{\braket}[2]{\langle#1|#2\rangle}

\newcommand{\brz}[1]{\langle \psi_0 | #1 | \psi_0 \rangle}

\newcommand{\KG}{K_\mathrm{good}}
\newcommand{\Aut}{\mathrm{Aut}}

\newcommand{\HH}[1]{(\mathcal{H}^2)^{\otimes #1}}
\newcommand{\HtoH}[1]{[\HH{#1} \to \HH{#1}]}

\newcommand{\NC}{\mathrm{NC}}

\newtheorem{thm}{Theorem}
\newtheorem{lemma}{Lemma}
\newtheorem{defn}{Definition}

\begin{document}
\large 
\maketitle

\begin{abstract}
  The notion of quantum hashing formalized by  F. Ablayev and  A. Vasiliev in 2013. F. Ablayev and M. Ablayev in 2014 introduced the notion of quantum hash generator which is convenient technical tool for constructing quantum hash functions. M. Ziatdinov in 2014 presented group approach for constructing quantum hash functions.  All these mentioned above results present constructions of quantum hash functions based on abelian groups.

  This paper  continue the research on quantum hashing. Our approach allows us to construct quantum hash function working on any (finite) group. Also our approach allows us to construct quantum hash functions based on classical hash function from $\NC^1$.

\end{abstract}

\section{Introduction}\label{sec:intro}

H. Buhrman et al.\ \cite{BCWdW01} introduce the notion of quantum fingerprinting.  Quantum fingerprinting  based on  binary error correcting codes. Later F. Ablayev and A. Vasiliev in~\cite{AV09} offer another (non binary) version of quantum fingerprinting. 
 F. Ablayev and A. Vasiliev \cite{AV14} defined notion of quantum hash-function and showed that quantum fingerprinting is a specific case of quantum hashing.

In~\cite{aa14} construction of Buhrman et al.\cite{BCWdW01} and Ablayev-Vasiliev's construction \cite{AV14} are generalized. It is shown that both approaches can be viewed as composition of so called ``quantum generator'' and (classical) universal hash function. 

 In \cite{Z2} we  offered  a group approach to fingerprinting. We showed  that instead of abelian group $\mathbb Z_m$ with   $m>0$ \cite{AV14} we can use arbitraey abelian group. These constructions use specific so called  ``good'' set of automorphisms. However, examples of such ``good'' sets (and, hence the quantum hash functions) were found only for abelian groups.

 In this paper we offer ``good'' set of automorphisms for symmetric group, and construct quantum hash function based on any finite group. This approach allows us to construct quantum hash functions based on classical functions from $\NC^1$. We also discuss the procedure of finding ``good'' set of automorphisms.

\section{Previous work}\label{sec:prior}

We start with recalling basic definitions that we will need in paper.

We will consider functions $h : \{0,1\}^n \to G$, where $G$ is a group.

Let us choose a set of automorphisms $\KK$ from group of all automorphisms $\Aut(G)$:
\begin{equation}\label{eqn:auto-set}
k_i \in \KK \subseteq \Aut(G), \qquad 1 \le i \le T, |\KK| = T
\end{equation}

We will use notation $k\{g\}$ for image of $g$ under automorphism $k$.

Let us also choose a homomorphism $f$ from group $G$ to a group of all unitary transformations of $m$ qubits.

\medskip

Let us recall definitions and theorems from \cite{AV14} and \cite{Z2}

Quantum hash function is defined as follows.
\begin{defn}
  $\ket{\Psi(w)}$ is a quantum hash function if it maps $n$--bit message $w$ from $\{0,1\}^n$ to $m$ qubits and resulting vectors are nearly orthogonal: $\forall w, w' \in \{0,1\}^n ( |\braket{\Psi(w)}{\Psi(w')}| < \epsilon )$ for some $\epsilon \in (0,1)$.
\end{defn}

We call set $\KG$ of elements of chosen $\KK$ ``good'' set if
for each non-unit group element $g$ and some starting state $\ket{\psi_0}$:
\begin{equation}\label{eq:good-set}
\forall g \in G, g \neq e : \frac{1}{|\KG|^2} \left| \sum_{k \in \KG} \brz{f(k\{g\})} \right|^2 < \epsilon
\end{equation}

In \cite{Z2} it was proved that
\begin{thm}\label{thm:good-set}
  If (\ref{eq:sum-good}) holds, then ``good'' set exists and can be constructed by choosing $d$ times element from $\KK$ at random, and $d = \frac{2}{\epsilon} \ln |G|$
  \begin{equation}\label{eq:sum-good}
    \forall g \in G, g\neq e : \frac{1}{|\KK|} \sum_{k \in \KK} \brz{f(k\{g\})} = 0,
  \end{equation}
\end{thm}
so, if (\ref{eq:sum-good}) holds, there exists quantum hash function for arbitrary small $\epsilon$ (however, ``good'' set size $d$ and therefore qubit count $m$ will grow)

We will say ``quantum hash function works for group $G$'' or simply ``quantum hash function for group $G$'' if it has form
\begin{equation}\label{eq:qhf-gr}
  \ket{\Psi_{h,G,K,f,m,\ket{\Psi_0}}(x)} = \frac{1}{\sqrt t} \sum_{j=0}^{t-1} \bigg( \ket{j} \otimes f\big( k_j\{h(x)\} \big) \ket{\psi_0} \bigg),
\end{equation}
where $h$ is classical hash function mapping $X^n$ to group $G$, $K = \{k_0, \ldots, k_{t-1}\}$  is ``good'' set of automorphisms and $f$ is homomorphism from $G$ to space $\HtoH{m}$.

It was also proven that
\begin{thm}\label{thm:qhf-exist}
  If for group $G$ ``good'' set of automorphisms exist, then quantum hash function for group $G$ exist.
\end{thm}

\section{Quantum hash function working on symmetric group}\label{sec:symm}

\begin{thm}\label{thm:symm-qhf}
  There exists a quantum hash function $\ket{\Psi_{h,S_n,K,f,\log n}}$ working on symmetric group.

  Specifically, $f$ is standard symmetric group representation in a space of $n$ dimensions and $K$ is a set of all automorphisms acting by conjugation to cyclic shift.
\end{thm}

\begin{proof}
  Theorems \ref{thm:good-set} and \ref{thm:qhf-exist} state that if there exists a homomorphism $f$, a set $\KK$ of automorphisms of $G$ such that 
  \begin{equation}\label{eq:good-qhf}
    \frac{1}{|\KK|} \sum_{k \in \KK} \brz{f(k\{g\})} = 0,
  \end{equation}
  then $\Psi_{h,G,K,f,m}$ is a quantum hash function.

  In our case, $f$ is a standard symmetric group representation in a space of $n$ dimensions with group $S_n$ acting by coordinates permutation.

  Let $\KK$ be the set of all (inner) automorphisms that has form:
  \begin{equation}\label{eq:def-kk}
    \KK = \{ g_\sigma : \sigma \text{ is a cyclic shift} \}, \quad g_\sigma(\tau) = \sigma \tau \sigma^{-1}\}
  \end{equation}

  Let $\ket{\psi_0}$ be some vector $c_1\ket{1} + c_2\ket{2} + \ldots + c_n\ket{n}$, such that:
  \begin{equation}\label{eq:psi0}
    \sum_{i=1}^n c_i = 0
  \end{equation}

  Image of $\ket{\psi_0}$ under $f(g_\tau\{\sigma\})$ for any $\sigma$ and $\tau \in \KK$ is
  \begin{equation}\label{eq:psi0-fg}
    f(g_\tau\{\sigma\}) = c_{\sigma(1+k)-k}\ket{1} + \ldots + c_{\sigma(n+k)-k}\ket{n},
  \end{equation}
  where $\tau$ is cyclic shift to $k$ and addition and subtraction in indices are modulo $n$.

  So, if we sum this for all automorphisms $\tau \in \KK$ we get:
  \begin{equation}\label{eq:sum-1}
    \sum_{g_\tau\in\KK}\brz{f(g_\tau\{\sigma\})} = \sum_{k=0}^n \sum_{i=0}^n c_i c_{\sigma(i+k)-k} = \sum_{i=0}^n c_i \sum_{k=0}^n c_{\sigma(i+k)-k}.
  \end{equation}
  We substituted $f(g_\tau\{\sigma\})\ket{\psi_0}$ with its value from \ref{eq:psi0-fg}.

  We can observe that $\sigma(i+k)-k$ runs over all integers from $1$ to $n$. So we can rewrite as follows:
  \begin{equation}\label{eq:sum-2}
    \sum_{g_\tau\in\KK}\brz{f(g_\tau\{\sigma\})} = \sum_{i=0}^n c_i \sum_{j=0}^n c_j = 0.
  \end{equation}
  We use equation (\ref{eq:sum-1}) and definition (\ref{eq:psi0}) of $\psi_0$.

  Equation (\ref{eq:sum-2}) is equivalent to (\ref{eq:good-qhf}), so theorems \ref{thm:good-set} and \ref{thm:qhf-exist} can be applied, and quantum hash function for $S_n$ exist.
\end{proof}

Please note that this proof does not apply to $A_5$ representation from paper \cite{AMP02} and we cannot use their representation and approach of this article to define quantum hash functions based on $NC^1$ functions. In the section \ref{sec:applications} we use another representation.

\bigskip

In \cite{Z2} it was shown that if we find a set $\KK$ satisfying equation (\ref{eq:good-set}), we can construct a ``good'' set with probability of $\frac{1}{|G|}$ by repeatedly ($d = \frac{2}{\epsilon} \ln |G|$ times) randomly choosing an element from $\KK$.

\section{Applications}\label{sec:applications}

We can use defined quantum hash function working on symmetric group to construct other quantum hash functions. One way of such construction is defined in \cite{Z2}: we construct a hash function working on (direct) product of groups. We present another way.

\begin{lemma}
  Let $G$ be a finite group, group $G' \lhd G$ be its subgroup, and $\ket{\Psi_{h,G,K,f,m}}$ be a quantum hash function working on it.

  Then we can define a quantum hash function working on $G'$.
\end{lemma}

\begin{proof}
We can define $h'$ to be a restriction of $h$ on $G'$.

Then $\ket{\Psi_{h',G',K,f,m}}$ is a quantum hash function.

Let us consider square of scalar product of quantum hash function values on different inputs.

\[
  \left| \braket{\Psi_{h',G',K,f,m}(x)}{\Psi_{h',G',K,f,m}(x)} \right|^2 = \left| \braket{\Psi_{h,G,K,f,m(x)}}{\Psi_{h,G,K,f,m}(x)} \right|^2 < \epsilon
\]

We use that $G' \lhd G$ and that $h'$ is a restriction of $h$ on $G'$.
\end{proof}

Of course, such way is inefficient for small finite subgroups of $S_n$, but it works for non-abelian groups.

\bigskip

We can use our approach to construct quantum hash functions based on classical hash functions in $\NC^1$.

Let $h$ be a hash function that can be computed by $\NC^1$ circuit. We can now use theorem \ref{thm:symm-qhf} to obtain a quantum hash function based on it as follows.

We can convert circuit to width--5 polynomial--size branching program and represent it as permutation branching program \cite{Bar89}. Then we compute quantum hash function based on $h$ as follows. For each input symbol we simultaneously apply required permutation in all subspaces (under different automorphisms as described in theorem \ref{thm:symm-qhf}).

\end{document}